\documentclass{ifacconf}

\usepackage{amsmath,amsfonts,amssymb}
\usepackage{graphicx}      
\usepackage{natbib}        

\usepackage{float}
\usepackage[usenames]{color}
\usepackage{comment}

\usepackage{mathtools}
\usepackage{multirow}

\usepackage{enumitem}

\usepackage{balance}

\usepackage{algorithm}

\usepackage{colortbl}

\newcommand\calc{\ensuremath{\mathcal C}}

\newcommand\kset{\ensuremath{\mathbb K}}

\newcommand{\R}{\mathbb{R}}

\newcommand{\N}{\mathbb{N}}

\newcommand{\argmin}{\operatorname{argmin}}

\newcommand{\0}{{0}}

\newcommand{\1}{{\bf 1}}

\newcommand{\fkz}{{\mathfrak{z}}}

\newcommand{\<}{\leqslant}
\renewcommand{\>}{\geqslant}

\newtheorem{lemma}{Lemma}
\newtheorem{theorem}{Theorem}

\newtheorem{remark}{Remark}
\newtheorem{corollary}{Corollary}
\newtheorem{definition}{Definition}

\newtheorem{assumption}{Assumption}

\begin{document}
\begin{frontmatter}

\title{Robust predictive control design for uncertain discrete switched affine systems subject to an input delay \thanksref{footnoteinfo}} 

\thanks[footnoteinfo]{The work of all authors was supported by grants PID2023-150118OB-I00 and ATR2023-145067 funded by MICIU/ AEI /10.13039/501100011033/.}

\author[First]{Gerson Portilla} 
\author[First]{Carolina Albea}
\author[First]{Alexandre Seuret} 

\address[First]{Departamento de Ingeniería de Sistemas y Automática de la Universidad de Sevilla (e-mail: gportilla,albea,aseuret@us.es)}

\begin{abstract}               
Robust stabilization conditions for uncertain switched affine systems subject to a unitary input delay are presented. They are obtained through the Lyapunov framework and a min-switching state-feedback predictive control law. The result relies on a prediction scheme considering nominal system parameters. By constructing a Lyapunov function that considers the prediction error, we demonstrate the exponential convergence of the system trajectories and system prediction to a robust limit cycle. An example is provided to validate the obtained result.
\end{abstract}


\end{frontmatter}

\section{Introduction}
Switched affine systems represent a specific category of hybrid systems characterized by a collection of affine subsystems and a switching signal that determines the active subsystem at any given time. In continuous time, the state evolves according to an affine differential equation, whereas in discrete time, it evolves according to an affine difference equation \citep{sun2011stability}. Moreover, its equilibrium is characterized by the convergence of a set of possibly several equilibrium points owing to the affine term \citep{deaecto2016stability}. It is worth mentioning that this class of systems is a crucial aspect of control theory, particularly in the context of power electronics \citep{theunisse2015robust,albea2015hybrid}.
\\
For the stabilization of discrete switched affine systems, the Lyapunov framework combined with periodic switching sequences has enabled the design of state-dependent control laws based on the min-switching strategy \citep{deaecto2016stability}. This methodology has allowed noticing the inherent convergence to an invariant limit cycle, rather than a fixed point, which is defined as a closed trajectory that includes equilibrium points and signifies a sustained oscillatory behavior \citep{rubensson1998convergence,rubensson2000stability}. The notion of limit cycle was further explored in \citet{serieye2023attractors}, where a precise definition and the necessary and sufficient conditions for its existence in switched affine systems were established.
\\
Switched affine systems have been studied in more complex scenarios, including uncertainties in system parameters, effects of additive disturbances, and controller input delays. In practice, the issue of input delay was exhibited in \citet{merchan2024data}, where the inherent delay between the power converter and the stabilizing switching signal required additional challenges for its stabilization.
\\
Robust control techniques for analysis and design \citep{scherer2001theory,ebihara2015s} have been effectively implemented, considering the system uncertainties and perturbations. For the continuous time domain, robust stabilization conditions expressed in terms of Linear Matrix Inequalities (LMI) along with min-switching control laws were developed in \citet{beneux2017robust} and \citet{beneux2019adaptive} for an unknown equilibrium point. For the discrete time domain, the analysis of exogenous inputs and uncertain parameters was addressed by \citet{egidio2020global} and \citet{serieye2023attractors}, respectively. In both cases, the convergence of the system solutions was ensured to a robust limit cycle, more precisely, a set of attractors around a limit cycle \citep{serieye2023attractors}.
\\
Regarding delayed controllers, delayed-state-feedback ones have been studied in \citet{mahmoud2010}, \citet{vu2010stability}, and \citet{mazenc2017state}, to mention a few. However, the stabilization of switched affine systems subject to a delayed control switching law reports few contributions in the literature \citep{sanchez2019switching,portilla2024}. In \citet{sanchez2019switching}, a class of MIMO bilinear systems with constant delays in both the state and the input switching control law, which is transformed into a switched affine one, was addressed via the Lyapunov-Krasovskii framework, resulting in a min-switching control law strategy that allowed ensuring the convergence to a given reference in the continuous time domain. In \citet{portilla2024}, assuming that the system parameters are perfectly known, the authors proposed a system-state-based predictive scheme to compensate for the input delay and designed a predictive switching control law that guarantees the global exponential stability of the closed-loop system to a limit cycle in the discrete time domain. 
\\
To the best of our knowledge, the stabilization of uncertain switched affine systems subject to an input delay in the control switching law has not been investigated thoroughly. This was initially discussed by \citet{portilla2024} only in simulations, observing that a predictive control law might induce the convergence to a neighborhood of the limit cycle. Unfortunately, for the case of uncertain system parameters, the direct extension of the developments presented in \cite{portilla2024} is not possible because of the uncertainty in the system parameters at each time, rendering it impossible to predict future states perfectly. \textcolor{black}{It is worth mentioning that the consideration of a predictor associates a completely different stabilization, which prevents the use of the stabilization conditions for the delay-free case in \citet{serieye2023attractors}.}  
\\
Based on the developments presented in \cite{portilla2024} for the case of input delay and the robust stabilization results obtained in \citet{serieye2023attractors} for the delay-free case, this contribution aims to present a robust predictive control design for uncertain switched affine systems subject to a unitary input delay, ensuring exponential convergence to a robust limit cycle. Here, we assume that the uncertain parameters of the system can be represented in a polytopic form. We construct a prediction scheme by selecting nominal parameters from this representation. This scheme allows us to synthesize a feedback control law based on min-switching techniques, compensating for the unitary input delay and uncertainties of the system parameters. A key aspect of our approach is the proposal of a Lyapunov function that incorporates the system state and past predictions. We demonstrate that both trajectories converge to an invariant set, specifically a robust limit cycle that aligns with this Lyapunov function. The main result is validated through an example.
\\
The paper is organized as follows. The definition of uncertain discrete switched affine systems subject to a unitary input delay, a previous result for the delay-free case, and the control objectives are introduced in Section~\ref{sec:form_prob}. In Section~\ref{sec:predic_control}, the prediction scheme and the main result are presented. An example validating our main result is discussed in Section~\ref{sec:example}. Finally, some comments and remarks are provided in Section~\ref{sec:conclusions}.
\\
\textbf{Notation:} The set of all real $n\times  m$ matrices is denoted by $\mathbb R^{n\times m}$, and the set of symmetric matrices in $\mathbb R^{n\times n}$ is denoted by $\mathbb S^n$. Matrices $I_n$ and $0_{n,m}$ ($0=0_{n,n}$) denote the identity matrix $\mathbb R^{n\times n}$ and null matrix $\mathbb R^{n\times m}$, respectively. When no confusion is possible, the subscripts of the matrices that specify the dimensions are omitted from the equations. For matrix $M$ of $\mathbb S^n$, the notation  $M\succ 0$ indicates that $M$ is positive definite. For any matrix $A=A^{\!\top},B,C=C^{\!\top}$ of appropriate dimensions, matrix $\left[\begin{smallmatrix}A&B\\\star & C  \end{smallmatrix} \right]$ denotes the symmetric matrix $\left[\begin{smallmatrix}A&B\\ B^{\!\top}& C  \end{smallmatrix} \right]$. 
$\text{Co}(A^{i},B^{i})$ refers to the convex hull generated by matrices $A^{i}$ and $B^{i}$, $i=1,\ldots,m$. For a matrix $P\succ 0$ and vector $x\in\R^n$, we denote $\|x\|_P=\sqrt{x^{\!\top} P x}$ as the weighted norm. For matrix $M\succ 0$ and vector $y\in\R^n$, we denote the shifted ellipsoid $\mathcal E(M,y)=\left\{x\in\mathbb R^n,~ (x-y)^{\!\top} M (x-y)\< 1 \right\}$.

\section{Problem formulation}\label{sec:form_prob}
\subsection{Uncertain switched affine systems with input-delay}
Consider the discrete switched affine system given by  
\begin{equation}\label{eq:model_x}
\begin{array}{lcl}
	x_{k+1}&=&A_{\sigma_{k-1}}x+B_{\sigma_{k-1}},\\
 \sigma_{k}&=& u_{\sigma,k},
\end{array}
\end{equation}
where $x_k\in\mathbb R^n$ is the system state, $\sigma_k \in \mathbb K:=\{1,2,\ldots,K\}$ is the switching signal with \textcolor{black}{$K$ the number of modes}, and $u_{\sigma,k}$ is the switching control law that selects the active mode in $\mathbb K$ for any time instant $k\in\N$. The initial conditions of the system \eqref{eq:model_x} are given by  $(x_0,\sigma_0,\sigma_{-1})$ in $\R^n\times \mathbb K^2$. Here, matrices $A_j\in\mathbb R^{n\times n}$ and $B_j\in\mathbb R^{n\times 1}$ for each mode $j \in \mathbb K=\{1,2,\dots,K\}$ define the system dynamics, which are assumed to be unknown and/or time-varying, but admit a polytopic representation given by
\begin{equation}\label{def:polytope}
    [A_{j},B_{j}]\in\text{Co}([A^{\ell}_{j},B^{\ell}_{j}])_{\ell\in\mathbb{L}},
\end{equation}
where $\mathbb{L}$ is a bounded subset of $\N$ and $A^{\ell}_{j}$ and $B^{\ell}_{j}$ are known and constant for any $j\in\mathbb K$ and any $\ell\in\mathbb L$. It is remarkable that the switching control law represented by $\sigma_{k-1}=u_{\sigma,k-1}$ is interpreted as a switching control law subject to a unitary delay. This approach contrasts with the work presented in \citet{serieye2023attractors} or \cite{egidio2020global}, where the switching control law is not delayed. \\
This work focuses on designing a suitable set-valued map $u_{\sigma,k}$ that robustly stabilizes uncertain switched affine systems subject to a unitary input delay to a set of vectors, specifically, a robust limit cycle to be characterized. It is worth mentioning that the delayed but nominal case was addressed in \citet{portilla2024}, arriving at the same stabilization condition obtained in \citet{serieye2023attractors} for the free-delay case. Here, from a state prediction using nominal system parameters, we provide a robust stabilization result for the uncertain system \eqref{eq:model_x}. 

\subsection{Hybrid limit cycles}
The adapted notion of a hybrid limit cycle and a robust limit cycle for discrete-time switched affine systems was presented in \cite{serieye2023attractors}, which is recalled in this section. Moreover, in the sequel, to highlight the difference between the nominal and uncertain system, \textcolor{black}{we denote $[\bar A_{j},\bar B_{j}]\in\text{Co}([A^{\ell}_{j},B^{\ell}_{j}])_{\ell\in\mathbb{L}}, j\in\mathbb{K},$ for some given nominal matrices of system \eqref{eq:model_x}}. To do so, we first introduce the set of periodic functions from $\mathbb N$ to $\mathbb K$: 
$$
		\calc\!=\!\left\{
  \nu \!:\! \mathbb {N} \rightarrow \kset,\mbox{ s.t. } \exists N\!\in\N\!\setminus\!\!\{0\},\forall \ell\in \mathbb {N},\ 
  \nu(\ell\!+\!N)\!=\!\nu(\ell)
  \right\}.
  $$
For a given periodic function $\nu \!\in\! \mathcal C$, $N_\nu$ and $\mathbb D_\nu\!=\!\{1,2,\dots, N_\nu\}$ denote the minimum period of $\nu$, i.e. the smallest integer such that $\nu(\ell\!+\!N_\nu)= \nu(\ell)$, $\forall\ell \in \mathbb N$, and the domain of $\nu$, respectively. We also introduce the modulo notation $\lfloor i\rfloor_{\nu}=((i-1) \textrm{ mod }  N_\nu)+1$, for any $i\in\mathbb N$, $i\> 1$. 
\begin{definition}[Hybrid limit cycle]\citep{serieye2023attractors}
\label{def:limitcycle}
A hybrid limit cycle for system~\eqref{eq:model_x}, or limit cycle in short, is a closed and isolated hybrid trajectory $s:\mathbb N\to\mathbb K \times \mathbb R^n$, $k\mapsto s_k=(\sigma_{k},x_k)$, which is a periodic (but not constant) solution of~\eqref{eq:model_x}.  
\end{definition}  
The notion of limit cycle presented in \citet{serieye2023attractors} indicates the existence and uniqueness of a sequence of vectors $\{\rho_i\}_{i\in\mathbb D_\nu}$ associated with a given $\nu$ in $\mathcal C$, satisfying 
\begin{equation}\label{def:rho_i}
    \rho_{\lfloor  i+1\rfloor_{\nu}}=\bar A_{\nu( i)}\rho_{ i}+\bar B_{\nu(i)},\quad \forall i\in\mathbb D_\nu,
\end{equation}
for some given and constant matrices $\bar A_j$ and $\bar B_j$, $\forall j\in\mathbb K$. However, as outlined in \cite{serieye2023attractors}, this notion cannot be guaranteed for system \eqref{eq:model_x} because the existence of a sequence of vectors $\{\rho_i\}_{i\in\mathbb D_\nu}$ is not unique in the case of uncertain and/or time-varying matrices $A_j$ and $B_j$, $\forall j\in\mathbb K$. 
However, the problem of uncertain and/or time-varying matrices $A_j$ and $B_j$, $\forall j\in\mathbb K$, can be treated by considering a neighborhood of the hybrid limit cycle, characterized by a subset of vectors $\mathcal{L}_i\subset\R^n$, $\forall i\in\mathbb D_\nu$, associated with $\nu\in\mathcal{C}$. Next, we introduce the concept of a robust limit cycle for system \eqref{eq:model_x}. 
\begin{definition}[Robust limit cycle]\citep{serieye2023attractors}\label{def:robust_limit}
    System \eqref{eq:model_x} admits a robust limit cycle associated with $\nu\in\mathcal{C}$, if there exist possibly disjoint subsets $\mathcal{L}_i\subset\R^n$, for $i\in\mathbb D_\nu$, such that $A_{\nu(i)}^{\ell}\mathcal{L}_i+B_{\nu(i)}^{\ell}\!\subset\! \mathcal{L}_{\lfloor  i+1\rfloor_{\nu}},\  \forall i\!\in\!\mathbb D_\nu,\  \forall\ell\!\in\!\mathbb L.$ For the case of known parameters, subsets $\mathcal{L}_i$ reduce to $\mathcal{L}_i=\rho_i,~\forall i\in\mathbb D_\nu$.
\end{definition}
\subsection{Preliminary result on the delay-free case}
In \cite{serieye2023attractors}, for the nominal and delay-free cases, that is, $\mathrm{Card}(\mathbb L)\!=\!1$ and $\sigma_k\!=\!u_{\sigma,k}$ (not $\sigma_{k+1}=u_{\sigma,k}$), the authors provided sufficient stabilization conditions for a limit cycle. This result is recalled, which is essential for obtaining stabilization conditions for system \eqref{eq:model_x}.
\begin{theorem}\citep{serieye2023attractors}\label{th:matrix_W}
For a given cycle $\nu$  in $\mathcal C$, assume that there exist matrices $\{P_i\}_{i\in\mathbb D_\nu}$ in $\mathbb S^n$, such that 
	\begin{equation}\label{eq:LMI}
		P_i\succ 0 , \quad \bar A_{\nu(i)}^{\!\top} P_{\lfloor i+1\rfloor_\nu} \bar A_{\nu(i)}\prec P_i,~\forall i \in\mathbb D_\nu.
	\end{equation}
	Then, the following statements hold:
\begin{enumerate}[label=(\alph*)]
           \item Cycle $\nu$ generates a unique hybrid nominal limit cycle for system~\eqref{eq:model_x}. More precisely, there exists a unique sequence of vectors $\{\rho_i\}_{i\in\mathbb D_\nu}$ that is a solution to the set of equations \eqref{def:rho_i}.
			\item $\mathcal A_\nu := \bigcup_{i\in\mathbb D_\nu}  \{\rho_i\}$ is globally exponentially stable (or equivalently is an attractor) for nominal and delay-free system \eqref{eq:model_x} with $A_i=\bar A_i$, $B_i=\bar B_i$, for all $i \in\mathbb D_\nu$, and $\sigma_k=u_{\sigma,k}$, and with the switching control law 
                \begin{equation}\label{eq:control_nominal}
            		u_{\sigma}(x_k) \!=\! \left\{\! \nu\left( \theta_k \right),\theta_k\!\in\!\underset{i\in\mathbb D_\nu}	{\argmin}\left(x_k\!-\!\rho_i\right)^{\!\top}\! P_i\!\left(x_k\!-\!\rho_i\right) \!\right\}.
            	\end{equation}
	        \item If  $\rho_i\neq \rho_j$ for all $i\neq j$ in $\mathbb D_\nu$ holds, then the switching signal $\sigma=u_{\sigma}(x)$ resulting from the closed-loop system \eqref{eq:model_x},\eqref{eq:control_nominal} converges ultimately to a shifted version of cycle $\nu$.
\end{enumerate}
\end{theorem}
\begin{remark}
    LMI \eqref{eq:LMI} was not introduced for the first time in \cite{serieye2023attractors}: It first appeared in the context of periodic systems in \cite{bolzern1988periodic}. In the same context of switched affine systems, the same condition was presented in \cite{egidio2020global} to prove the stability of a limit cycle using a different control law. 
\end{remark}
\begin{remark}\label{ass:monodromy}
	It is worth mentioning that condition \eqref{eq:LMI} of Theorem~\ref{th:matrix_W} is equivalent to verifying the Schur stability of the monodromy matrix $\Phi_{\nu}$, defined by $\Phi_{\nu}=\Pi_{\iota=1}^{N_{\nu}} \bar A_{\nu(\iota)}$.
\end{remark}
\begin{corollary}
    If inequalities \eqref{eq:LMI} holds, then there exists a sufficiently small $\mu$ in $(0,1)$ such that 
    \begin{equation}
\bar A_{\nu(i)}^{\!\top}P_{\lfloor i+1\rfloor_{\nu}}\bar A_{\nu(i)}\prec\! (1-\mu)   P_{i}, \quad \forall i\in\mathbb D_\nu.\label{eq:LMI_mu}
\end{equation}
\end{corollary}
\begin{pf}
The proof is easily derived from the finite number of strict inequalities in \eqref{eq:LMI}.
\end{pf}
\subsection{Autonomous augmented delay-free model}\label{subsec:delay_free_model}
As presented in \citet{portilla2024}, for a given $\nu$ in $\mathcal C$, system \eqref{eq:model_x} can be rewritten as follows
\begin{equation}\label{eq:model_xitheta}
	x_{k+1}=A_{\nu(\theta_k)}x_k+B_{\nu(\theta_k)},\ 
 \theta_{k+1}=  {u_k \in\mathbb D_\nu},
\end{equation}
where $u_k\!\in\!\mathbb D_\nu$ is the new control input, such that $u_{\sigma,k}\!=\!\nu(\theta_{k+1})$ belongs to $\mathbb K$. 
\textcolor{black}{Regarding the augmented state $\begin{bmatrix}x_k^\top&x_{k-1}^\top\end{bmatrix}^\top$, the system dynamics can be rewritten as \begin{equation}\label{eq:model_xitheta_output}
\begin{split}
\begin{array}{ccl}
	\begin{bmatrix}
	    x_{k+1}\\x_{k}
	\end{bmatrix}\!\!\!=\!\!\begin{bmatrix}
	    A_{\nu(u(y_k))}& \!\!0\\
     I&\!\! 0\\
	\end{bmatrix}\!\!\begin{bmatrix}
	    x_k\\x_{k-1}
	\end{bmatrix}\!\!+\!\!\begin{bmatrix}
	    B_{\nu(u(y_k))}\\
     0
	\end{bmatrix}\!,y_k\!=\!\begin{bmatrix}
0&I
	\end{bmatrix}\!\!\begin{bmatrix}
	    x_k\\x_{k-1}
	\end{bmatrix}
\end{array}\!,
\end{split}
\end{equation}
where the output vector $y_k$ indicates the available information for the control design, leading to the problem of output feedback control design of $u(y_k)$. Consequently, the method proposed in \cite{serieye2023attractors} is inapplicable, as it relies on a control Lyapunov function, which requires full knowledge of the state vector.}

\subsection{Control objectives}
To provide a robust stabilization criterion considering system uncertainties, we consider the nominal system \eqref{eq:model_xitheta} with known matrices $A_j=\bar A_j$, $B_j=\bar B_j$, $\forall j\in\mathbb K$, and hold the following assumption in the sequel.
\begin{assumption}\label{ass:nominal_system}
    For a given cycle $\nu\in\mathcal{C}$ and given matrices \textcolor{black}{$[\bar A_{j},\bar B_{j}]\in\text{Co}([A^{\ell}_{j},B^{\ell}_{j}])_{\ell\in\mathbb{L}}$}, there exist $\mu\in(0,1)$, $\{\rho_i\}_{i\in\mathbb D_\nu}$ in $(\mathbb R^n)^{N_\nu}$ and $\{P_i\}_{i\in\mathbb D_\nu}\in\mathbb S^n$, satisfying \eqref{def:rho_i} and \eqref{eq:LMI_mu}.
\end{assumption}
Under Assumption~\ref{ass:nominal_system}, our objectives are as follows: 
\begin{itemize}
    \item Design a predictive scheme for compensating a unitary input delay, using the nominal matrices $[\bar A_{j},\bar B_{j}]
    , j\in\mathbb{K}$. 
    \item Under Assumption~\ref{ass:nominal_system}, design a robust predictive state-feedback min-switching control law that stabilizes the uncertain system \eqref{eq:model_xitheta}.
    \item Under Assumption~\ref{ass:nominal_system}, estimate a set of attractors, expressed as level sets of a Lyapunov function determined by the nominal limit cycle $\{\rho_i\}_{i\in\mathbb D_\nu}$ and matrices $\{P_i\}_{i\in\mathbb D_\nu}$ in $\mathbb S^n$, which are associated with a given $\nu \in \mathcal C$. The level sets must contain the subsets $\mathcal{L}_i\subset\R^n$, $i\in\mathbb D_\nu$, introduced in Definition~\ref{def:robust_limit}, while the memory variables $\theta$ and $\vartheta$ remain in $\mathbb D_\nu$.  
 
\end{itemize}

\section{Robust predictive control design}\label{sec:predic_control}
\subsection{Definition and properties of the predictor}
 Inspired by our previous study on predictive control for delayed time-invariant systems \citep{portilla2024}, we introduce the following prediction scheme for system \eqref{eq:model_xitheta}:
    \begin{align}
        \chi_0=x_k,\quad 
    \chi_1=\bar A_{\nu(\theta_k)}\chi_0 + \bar B_{\nu(\theta_k)},\label{def:predictor_nominal}
    \end{align}
which aims to predict the future state $x_{k+1}$ through $\chi_1$. Note that the prediction scheme \eqref{def:predictor_nominal} can be implemented because the nominal matrices $\{\bar A_{j}, \bar B_{j}\}_{ j\in\mathbb{K}}$ are assumed to be known. \textcolor{black}{Using prediction \eqref{def:predictor_nominal}, we propose the following control input to compensate for the input delay:
\begin{equation}\label{eq:control}
    u_k\in\underset{i\in\mathbb D_\nu}{\argmin}\Big( \left\Vert \chi_1\!-\!\rho_{i}\right\Vert_{P_{i}}^2 \Big).
\end{equation}}\\
\textcolor{black}{Note that $\chi_{1}$ is build upon the knowledge $x_k$ and $\theta_k$. Recalling that $\theta_k$ is the previous control action, we have $\theta_k\in\argmin_{i\in\mathbb D_\nu}\big( \left\Vert \fkz_1\!-\!\rho_{i}\right\Vert_{P_{i}}^2 \big)$, where $\fkz_1$ refers to the prediction made at the previous step and defined by
\begin{align}
    \fkz_0=z_k,\quad 
    \fkz_1=\bar A_{\nu(\vartheta_k)}\fkz_0 + \bar B_{\nu(\vartheta_k)}.\label{def:predictor_nominal_past}
\end{align} 
Here, $\fkz_0$ and  $\vartheta_k$ stand for the memory variables of $x_k$ and $\theta_k$, respectively, obtained at the previous step. More precisely, we have $z_{k+1}=x_k$ and $\vartheta_{k+1}=\theta_k$.} 
\textcolor{black}{
Due to the uncertainties in the systems matrices, the prediction at the previous step $\fkz_1$ and the current state $x_k$ are possibly different. Therefore, we need to introduce the complete state vector of the closed-loop system
\begin{equation}
    \xi_k:=\begin{bmatrix}x_k^{\!\top} & z_k^{\!\top} & \theta_k & \vartheta_k \end{bmatrix}^{\!\top} \in\mathbb H^n:= \mathbb R^{2n} \!\times\! \mathbb D^2_\nu,
\end{equation}
Besides, we adopt the following notation in the sequel: $\xi^+=\xi_{k+1}$, $\xi=\xi_{k}$, $x=x_{k}$, $z=z_{k}$, $\theta=\theta_{k}$ and $\vartheta=\vartheta_{k}$. To emphasize the construction of the predictors upon the state vector $\xi$, we will also use notation $\chi_s(\xi)$ and $\fkz_s(\xi)$, for $s=0,1$.} 
\\
\textcolor{black}{As in \cite{portilla2024}, we can state the following lemma linking the current prediction and the prediction made at the previous step.}
\begin{lemma}\label{lem:pred}
   The predictors \eqref{def:predictor_nominal} and \eqref{def:predictor_nominal_past} verify $\fkz_s(\xi^+)=\chi_s(\xi),\  \forall s=0,1$, for any trajectories $(\xi^+,\xi)$.
\end{lemma}
\begin{pf}
    Consider any trajectory $(\xi^+,\xi)$. For $s=0$, the dynamic of $\xi$ guarantees $\fkz_0(\xi^+)=z^+=x=\chi_0(\xi)$. Consequently, for $s=1$, predictor \eqref{def:predictor_nominal_past} yields 
    \begin{equation*}
        \fkz_1(\xi^+)=\bar A_{\nu(\theta)}\fkz_0(\xi^+) + \bar B_{\nu(\theta)}=\bar A_{\nu(\theta)}\chi_0(\xi) + \bar B_{\nu(\theta)}=\chi_1(\xi).
    \end{equation*}
\end{pf}
\textcolor{black}{The main novelty with respect to \cite{portilla2024} comes from the fact that the uncertainties bring mismatch between the previous prediction and the current state. Therefore, we need to state the essential lemma, which quantifies this mismatch.}
\begin{lemma}\label{lemma:pred_error}
    Under Assumption~\ref{ass:nominal_system}, the identity 
    \begin{equation}\label{eq:chi_0plus}
    \chi_0(\xi^+)-\chi_1(\xi) =\Delta A_{\nu(\theta)}(\chi_0\!-\!\rho_{\theta})   + \delta_{\theta}
    \end{equation}
holds for any trajectory $(\xi^+,\xi)$, where $\Delta A_{\nu(\theta)}=A_{\nu(\theta)}-\bar A_{\nu(\theta)}$ denote the uncertainty in the state matrix and where \begin{equation}
    \delta_{\theta}=A_{\nu(\theta)}\rho_{\theta}+ B_{\nu(\theta)} -\rho_{\lfloor \theta+1\rfloor_{\nu}}, \quad \theta\in\mathbb D_\nu \label{def:delta}
\end{equation} is the deviation of the nominal limit cycle due to the uncertainties in the system matrices.
\end{lemma}
\begin{pf}
Note that, for any trajectory $\xi\in\mathbb H^{n}$, the identity 
\begin{align*}
   \chi_0(\xi^+)-\chi_1(\xi)&= A_{\nu(\theta)}x+ B_{\nu(\theta)}-\bar A_{\nu(\theta)}\chi_0(\xi) - \bar B_{\nu(\theta)}\nonumber
\end{align*}
 holds. Then, recalling the definition of the components of the limit cycle $\{\rho_i\}_{i\in\mathbb D}$ satisfying \eqref{def:rho_i} and taking into account that $x=\chi_0(\xi)$, we finally obtain
\begin{align*}
   \chi_0(\xi^+)-\chi_1(\xi)&= \Delta A_{\nu(\theta)}\chi_0(\xi)\!+\!B_{\nu(\theta)}\!-\!\bar B_{\nu(\theta)}\!\pm\! \Delta A_{\nu(\theta)}\rho_{\theta}\nonumber\\
  &= \Delta A_{\nu(\theta)}(\chi_0\!-\!\rho_{\theta})   + \delta_{\theta}.
\end{align*}
\end{pf}
\textcolor{black}{Note that when there is no uncertainty, we retrieve $\chi_0(\xi^+)-\chi_1(\xi)=0$ as assumed in \cite{portilla2024}.} Next, an essential property of the predictor \eqref{def:predictor_nominal} is presented, exploiting condition \eqref{eq:LMI_mu} of Theorem~\ref{th:matrix_W}.
\begin{lemma}\label{lem:bound_pred_chi_1}
    Under Assumption~\ref{ass:nominal_system}, it holds
    \begin{equation}\label{eq:chi_1plus}
    \Vert \chi_1(\xi)-\rho_{\lfloor  \theta \!+\!1\rfloor_{\nu}}\! \Vert _{P_{\lfloor  \theta +\!1\rfloor_{\nu}}}^2 \!\!\< \!(1-\mu)  \Vert \chi_0(\xi)-\!\rho_{\theta} \Vert _{P_{\theta }}^2,\ \forall \xi\!\in\!\mathbb H^{n}\!.
    \end{equation}
\end{lemma}
\begin{pf}
If Assumption~\ref{ass:nominal_system} holds, then there exist vectors $\{\rho_i\}_{i\in\mathbb D_\nu}$ and matrices $\{P_i\}_{i\in\mathbb D_\nu}$ in $\mathbb S^n$, satisfying \eqref{def:rho_i} and \eqref{eq:LMI_mu}, respectively. Then, the definition of predictor \eqref{def:predictor_nominal} and equation \eqref{def:rho_i} yields
\begin{align*}
    \chi_1(\xi)\!-\!\rho_{\lfloor\theta+1\rfloor_{\nu}}\!&=\!\bar A_{\nu(\theta)}\chi_0(\xi)+\bar B_{\nu(\theta)}-\rho_{\lfloor  \theta+1\rfloor_{\nu}}\nonumber\\
    &=\bar A_{\nu(\theta)}\chi_0(\xi)+\bar B_{\nu(\theta)}-\bar A_{\nu(\theta)}\rho_{\theta}-\bar B_{\nu(\theta)}\nonumber\\
    &=\bar A_{\nu(\theta)}(\chi_0(\xi)-\rho_{\theta}).
\end{align*}
Computing now the weighted norm of $\chi_1(\xi)\!-\!\rho_{\lfloor  \theta+1\rfloor_{\nu}}$ with matrix $P_{\lfloor  \theta+1\rfloor_{\nu}}$ and from \eqref{eq:LMI_mu}, it yields
\begin{align}
\Vert \chi_1(\xi)-\rho_{\lfloor  \theta+1\rfloor_{\nu}} \Vert _{P_{\lfloor  \theta+1\rfloor_{\nu}}}^2&\!\!=\Vert \bar A_{\nu(\theta)}(\chi_0(\xi)-\rho_{\theta} )\Vert _{P_{\lfloor  \theta+1\rfloor_{\nu}}}^2\nonumber\\
&\!\!=\Vert (\chi_0(\xi)-\rho_{\theta} )\Vert _{\bar A_{\nu(\theta)}^{\!\top} P_{\lfloor  \theta+1\rfloor_{\nu}}\bar A_{\nu(\theta)}}^2\!\!\nonumber\\
&\!\!<\Vert (\chi_0(\xi)-\rho_{\theta} )\Vert _{(1-\mu) P_{\theta}}^2\nonumber,\label{ineq:chi10}
\end{align}
which concludes the proof.
\end{pf}
\subsection{Main result}
We are now able to state the main result of this paper.
\begin{theorem}\label{th:theo3}
Consider given cycle $\nu$  in $\mathcal C$, and for any given matrices $[\bar A_{j},\bar B_{j}]\in\text{Co}([A^{\ell}_{j},B^{\ell}_{j}])_{\ell\in\mathbb{L}},$ $j \in \mathbb K$ such that Assumption~\ref{ass:nominal_system} holds, i.e., there exist $\mu\in(0,1)$ and $\{\rho_i,P_i\}_{i\in\mathbb D_\nu}$ satisfying \eqref{def:rho_i}, \eqref{eq:LMI_mu}. Assume that there exist,  for a given parameter $\gamma\in(0,1)$, two matrices  $0\preceq R\in\mathbb S^n$ and $0\prec Q\in\mathbb S^n$, such that
\begin{equation}\label{eq:condition}
    \begin{bsmallmatrix}
        R+Q & -Q\\
        \star & P_{i}+Q-2R 
    \end{bsmallmatrix}\!\succ\! 0,\
    \Psi_i(A_{\nu(i)}^\ell,B_{\nu(i)}^\ell)\succ 0,\ \forall i \in\mathbb D_\nu,~ \ell\in \mathbb L,
\end{equation}
where
\begin{multline*}
    \Psi_i(A_{\nu(i)}^\ell,B_{\nu(i)}^\ell)=\\\begin{bsmallmatrix}
        (1\!-\!\gamma)(R+Q)-(1-\mu)P_{i} & -(1\!-\!\gamma)Q & 0 & (A_{\nu(i)}^\ell \!-\! \bar A_{\nu(i)})^{\!\top} (Q\!+\!2R)\\
        \star & (1\!-\!\gamma)(P_{i}+Q-2R) & 0 & 0\\
        \star & \star & \gamma & \delta_{i}^{\ell\top}\!(Q+2R)\\
        \star & \star & \star & Q+2R
    \end{bsmallmatrix},
\end{multline*}
and
$$\delta_{i}^\ell=A_{\nu(i)}^\ell\rho_{i} +   B_{\nu(i)}^\ell -\rho_{\lfloor i+1\rfloor_{\nu}}.$$
Then, \textcolor{black}{$$\mathcal S_\nu := \Big\{\xi\in\mathbb H^{n},\mbox{ s.t. }V(\xi)\<1,~\theta\in
    \underset{i\in\mathbb D_\nu}{\argmin}\Big( \left\Vert \fkz_1\!-\!\rho_{i}\right\Vert_{P_{i}}^2\Big)\Big\},$$}
where
$$V(\xi)= \Vert \fkz_1\!-\!\rho_{\theta} \Vert_{P_{\theta}-2R}^2  + \Vert \chi_0\!-\!\fkz_1 \Vert_{Q}^2+\Vert \chi_0\!-\!\rho_{\theta} \Vert_{R}^2,$$
is robustly globally exponentially stable and \textcolor{black}{invariant} for system \eqref{eq:model_xitheta} with the predictive switching control law \eqref{eq:control}. 
\end{theorem}
\begin{pf}
For the sake of readability, we will first use the short-hand notations $\chi_s=\chi_s(\xi)$, $\fkz_s=\fkz_s(\xi)$ $\chi_s^+=\chi_s(\xi^+)$, and $\fkz_s^+=\fkz_s(\xi^+)$ for $s=0,1$. 

The proof consists of proving that the current state $x$ converges to a neighborhood of the previous prediction $\fkz_1$ determined by the nominal limit cycle $\{\rho_i\}_{i\in\mathbb D_\nu}$, thanks to the closed-loop system \eqref{eq:model_xitheta}-\eqref{eq:control}. This suggests the construction of a Lyapunov function that measures the distance of the previous prediction $\fkz_1$, the current state $\chi_0$ and the prediction error $\chi_0-\fkz_1$ to the nominal limit cycle $\{\rho_i\}_{i\in\mathbb D_\nu}$, verifying that this function decreases along the solutions of the closed-loop system \eqref{eq:model_xitheta}-\eqref{eq:control}. To do so, we \textcolor{black}{regard the trajectories $\xi\in\mathbb H^n$ satisfying $\theta\in
    \underset{i\in\mathbb D_\nu}{\argmin}\Big( \left\Vert \fkz_1\!-\!\rho_{i}\right\Vert_{P_{i}}^2\Big)$} and propose the candidate Lyapunov function given by
\begin{equation}\label{eq:LF}
V(\xi)\!=\Vert \fkz_1\!-\!\rho_{\theta} \Vert_{P_{\theta}-2R}^2  + \Vert \chi_0\!-\!\fkz_1 \Vert_{Q}^2+\Vert \chi_0\!-\!\rho_{\theta} \Vert_{R}^2,~\forall\xi\in\mathbb H^n,
\end{equation}
where vectors $\{\rho_i\}_{i\in\mathbb D_\nu}$ and matrices $\{P_i\}_{i\in\mathbb D_\nu}$ in $\mathbb S^n$ are solution to \eqref{def:rho_i} and \eqref{eq:LMI_mu} (Assumption~\ref{ass:nominal_system}), respectively. Moreover, it is clear that $V(\xi)>0$ due to \eqref{eq:condition}. 

The forward increment of $V(\xi)$ writes
\begin{align}
	\Delta V(\xi)&\!=\!\displaystyle \Vert \fkz_1^+\!-\!\rho_{\theta^+}\Vert_{P_{\theta^+}\!\!-2R}^2 \!+\! \Vert \chi_0^+\!-\!\fkz_1^+ \Vert_{Q}^2 \!+\!\Vert \chi_0^+\!-\!\rho_{\theta^+} \Vert_{R}^2 \nonumber\\
	&-\Vert \fkz_1\!-\!\rho_{\theta} \Vert_{P_{\theta}-2R}^2  - \Vert \chi_0\!-\!\fkz_1 \Vert_{Q}^2-\Vert \chi_0\!-\!\rho_{\theta} \Vert_{R}^2\nonumber.
\end{align}
Now, using Lemma~\ref{lem:pred} and recalling that $\theta^+=u$, it yields
\begin{align}
	\Delta V(\xi)&\!=\!\displaystyle \Vert \chi_1\!-\!\rho_{u}\Vert_{P_{u}-2R}^2 + \Vert \chi_0^+\!-\!\chi_1 \Vert_{Q}^2 +\Vert \chi_0^+\!-\!\rho_{u} \Vert_{R}^2 \nonumber\\
	&-\Vert \fkz_1\!-\!\rho_{\theta} \Vert_{P_{\theta}-2R}^2  - \Vert \chi_0\!-\!\fkz_1 \Vert_{Q}^2-\Vert \chi_0\!-\!\rho_{\theta} \Vert_{R}^2\nonumber.
\end{align}
Focusing on the third term,
\begin{align*}
	\Vert \chi_0^+\!-\!\rho_{u}\pm\chi_1 \Vert_{R}^2&=\Vert \chi_0^+\!-\!\chi_1 \Vert_{R}^2 + \Vert \chi_1\!-\!\rho_{u} \Vert_{R}^2\\
    & + 2(\chi_0^+\!-\!\chi_1)^{\!\top} R (\chi_1\!-\!\rho_{u}) ,
\end{align*}
which, using Young's inequality for the last term, admits the following upper bound
\begin{equation*}
	\Vert \chi_0^+\!-\!\rho_{u} \Vert_{R}^2\<\Vert \chi_0^+\!-\!\chi_1 \Vert_{2R}^2 + \Vert \chi_1\!-\!\rho_{u} \Vert_{2R}^2.
\end{equation*}
Thus, the forward increment of $V(\xi)$ is bounded as follows
\begin{align}
	\Delta V(\xi)\!&\<\!\displaystyle \Vert \chi_1\!-\!\rho_{u}\Vert_{P_{u}}^2 + \Vert \chi_0^+\!-\!\chi_1 \Vert_{Q+2R}^2\nonumber\\
	&-\Vert \fkz_1\!-\!\rho_{\theta} \Vert_{P_{\theta}-2R}^2  - \Vert \chi_0\!-\!\fkz_1 \Vert_{Q}^2-\Vert \chi_0\!-\!\rho_{\theta} \Vert_{R}^2\nonumber.
\end{align}
As $u$ is defined by the control law in \eqref{eq:control}, the following inequality $\Vert \chi_1\!-\!\rho_{u}\Vert_{P_{u}}^2 \< \Vert \chi_1\!-\!\rho_{\iota}\Vert_{P_{\iota}}^2$ holds for any $\iota$ in $\mathbb D_\nu$.  Therefore, selecting $\iota =\lfloor \theta+1\rfloor_{\nu}$, we obtain 
\begin{align}
	\Delta V(\xi)\!&\<\!\displaystyle \Vert \chi_1\!-\!\rho_{\lfloor \theta+1\rfloor_{\nu}}\Vert_{P_{\lfloor \theta+1\rfloor_{\nu}}}^2 + \Vert \chi_0^+\!-\!\chi_1 \Vert_{Q+2R}^2\nonumber\\
	&-\Vert \fkz_1\!-\!\rho_{\theta} \Vert_{P_{\theta}-2R}^2  - \Vert \chi_0\!-\!\fkz_1 \Vert_{Q}^2-\Vert \chi_0\!-\!\rho_{\theta} \Vert_{R}^2\nonumber.
\end{align}
Using Lemmas~\ref{lemma:pred_error} and \ref{lem:bound_pred_chi_1}, we derive an upper bound of $\Delta V(\xi)$ as follows
\begin{align}
	\Delta V(\xi)\!&\<\!\displaystyle \Vert \chi_0-\rho_{\theta} \Vert_{(1-\mu)P_{\theta}} ^2+ \Vert \Delta A_{\nu(\theta)}(\chi_0\!-\!\rho_{\theta})   + \delta_{\theta} \Vert_{Q+2R}^2\nonumber\\
	&-\Vert \fkz_1\!-\!\rho_{\theta} \Vert_{P_{\theta}-2R}^2  - \Vert \chi_0\!-\!\fkz_1 \Vert_{Q}^2-\Vert \chi_0\!-\!\rho_{\theta} \Vert_{R}^2\nonumber,
\end{align}
or equivalently,
\begin{equation*}
   \begin{split}
     \Delta V(\xi)\!&\<-\mathcal{\zeta}_{\theta}^{\!\top}\left(\begin{bsmallmatrix}
        R+Q-(1-\mu)P_{\theta} & -Q & 0\\
        \star & P_{\theta}+Q-2R & 0\\
        \star & \star & 0\end{bsmallmatrix}\right.\\
        &- \left.\begin{bsmallmatrix}
        \Delta A_{\nu(\theta)}^{\!\top}\\0 \\\delta_{\theta}\end{bsmallmatrix} (Q+2R)\begin{bsmallmatrix}
        \Delta A_{\nu(\theta)}^{\!\top}\\0 \\ \delta_{\theta}\end{bsmallmatrix}^{\!\top} \right)\mathcal{\zeta}_{\theta},
   \end{split} 
\end{equation*}
where $\mu\in(0,1)$ is a fixed parameter, $\delta_\theta$ as in \eqref{def:delta}
and
\begin{equation}
{\zeta}_{\theta}=\begin{bmatrix}
(\chi_0\!-\!\rho_{\theta})^{\!\top} & (\fkz_1\!-\!\rho_{\theta})^{\!\top} & 1\end{bmatrix}^{\!\top}.\label{def_zeta}
\end{equation}
Recall that the objective is to ensure the asymptotic stability of $\mathcal S_\nu$, i.e. $V(\xi)\> 1$.  Using the augmented vector $\zeta_\theta$, we note that 
$$V(\xi)\> 1 \quad \Leftrightarrow \quad \zeta_\theta^{\!\top}\begin{bsmallmatrix}
        R+Q & -Q & 0\\
        \star & P_{\theta}+Q-2R & 0\\
        \star & \star & -1
    \end{bsmallmatrix}\zeta_\theta\>0.$$
Therefore, condition $\Delta V(\xi) < 0$ for any $\xi\not\in\mathcal S_\nu$ is equivalent, using an S-procedure, to the existence of a positive scalar $\gamma\in(0,1)$ such that
\begin{equation*}
	\begin{split}
		&\begin{bsmallmatrix}
        R+Q-(1-\mu)P_{\theta} & -Q & 0\\
        \star & P_{\theta}+Q-2R & 0\\
        \star & \star & 0\end{bsmallmatrix}- \begin{bsmallmatrix}
        \Delta A_{\nu(\theta)}^{\!\top}\\0 \\\delta_{\theta}\end{bsmallmatrix} (Q+2R)\begin{bsmallmatrix}
        \Delta A_{\nu(\theta)}^{\!\top}\\0 \\ \delta_{\theta}\end{bsmallmatrix}^{\!\top} \\
		&-\gamma \begin{bsmallmatrix}
        R+Q & -Q & 0\\
        \star & P_{\theta}+Q-2R & 0\\
        \star & \star & -1
\end{bsmallmatrix}\succ 0.
	\end{split} 
\end{equation*}
With the help of the Schur complement, the previous matrix inequality is also equivalent to $\Psi_{\theta}(A_{\nu(\theta)},B_{\nu(\theta})$. 
The last step consists of recalling that matrices $[A_{j}\ B_{j}]$, $j\in\mathbb{K}$, admit the polytopic representation in \eqref{def:polytope}, and noting that matrix  $\Psi_{\theta}$ is affine and consequently, convex with respect to \eqref{eq:condition} is affine with respect to $(A_{\nu(\theta)},B_{\nu(\theta})$.\\ 
To complete the proof, it remains to show that $\mathcal S_\nu$ is invariant. If condition \eqref{eq:condition} is satisfied, we already showed that $\Delta V(\xi) + \mu ( V(\xi)-1)<0$, which implies that 
\begin{equation*}
    \begin{split}
        V(\xi^+)&=V(\xi)-\gamma ( V(\xi)-1)\<(1-\gamma)V(\xi)+\gamma.
    \end{split}
\end{equation*}
Since $\xi\in \mathcal S_\nu$ ($V(\xi)\leq 1$) and $\gamma\in(0,1)$, then $V(\xi^+)\<(1-\gamma)+\gamma=1$ holds true. Therefore, $\xi^+\in \mathcal S_\nu$.
\end{pf}
\begin{remark}
    The decision variables of condition \eqref{eq:condition} in Theorem~\ref{th:theo3} are matrices $R\in\mathbb S^n$ and $Q\in\mathbb S^n$, and parameter $\gamma\in(0,1)$, whereas $\{\rho_i,P_i\}_{i\in\mathbb D_{\nu}}$ and $\mu$ are fixed a priori by solving \eqref{def:rho_i} and \eqref{eq:LMI_mu}. Therefore, fixing $\gamma$ as a tuning parameter, condition \eqref{eq:condition} becomes an LMI.
\end{remark}
\begin{remark}
  Condition \eqref{eq:condition} depends only on the current mode, $\theta$, and not on the previous one $\vartheta$, which has been introduced only to provide a complete stability analysis. This relevant fact drastically reduces the complexity of the condition \eqref{eq:condition}, since one would have to consider all the possible combinations of couples $(\theta,\vartheta)$. This allows us to extend this setup to larger delays.
\end{remark}
\begin{remark}\label{remark:neighborhood}
    Following Definition~\ref{def:robust_limit}, the attractor $\mathcal{S}_{\nu}$ is defined with variables $(\chi_0,\fkz_1)\in\mathbb R^{2n}$, which makes it complicated to illustrate graphically. However, one may consider its projection onto the set, where  $\chi_0=\fkz_1(=x)$, which leads to the condition $V(\xi)\<1$ yielding $\Vert x\!-\!\rho_{\theta} \Vert_{P_{\theta}-R}^2\<1$. This allows us to estimate the robust limit cycle as $\mathcal{L}_i\subset\mathcal E(P_i-R,\rho_i)$, for all $i\in\mathbb D_\nu$.
\end{remark}
\begin{remark}
    \textcolor{black}{The feasibility of conditions \eqref{eq:condition} can be verified to the case $A_{\nu(i)}\approx \bar A_{\nu(i)}$ and $B_{\nu(i)}\approx \bar B_{\nu(i)}$, where a solution to the decision variables in Theorem~\ref{th:theo3} can be given by
    \begin{equation*}
        \gamma<\mu,\quad  R=0,\quad Q\succ \dfrac{1-\mu}{\mu-\gamma} P_i,  \quad \forall i\in\mathbb D_\nu.
    \end{equation*}}
\end{remark}

\section{Illustrative example}\label{sec:example}
Consider the discrete-time switched affine system \eqref{eq:model_xitheta} borrowed from \cite{egidio2019novel}, consisting of two unstable modes given by  
\begin{equation*}
    A_1(\delta_1)=e^{F_1 T}+\begin{bsmallmatrix}
            0 & \delta_1 & 0\\0 & 0 & 0\\0& 0 &0
        \end{bsmallmatrix},~B_1(\delta_1)=\int_0^{T} e^{F_1 s}g_1\textup{d}s  + \begin{bsmallmatrix}
            0\\4\delta_1\\0
        \end{bsmallmatrix},
\end{equation*}
\begin{equation*}
    A_2(\delta_2)=e^{F_2 T}+\begin{bsmallmatrix}
            0 & 0 & 0\\\delta_2 & 0 & 0\\0&0&0
        \end{bsmallmatrix},~B_2(\delta_2)=\int_0^{T} e^{F_2 s}g_2\textup{d}s  + \begin{bsmallmatrix}
            1.4\delta_2\\0\\0
        \end{bsmallmatrix},
\end{equation*}
where
\begin{equation*}
    F_1\!=\!\begin{bsmallmatrix}
            -3 & -6 & 3\\2 & 2 & -3\\1.6& 0 &-2
        \end{bsmallmatrix},~F_2\!=\!\begin{bsmallmatrix}
            1 & 3 & 3\\-0.2 & -3 & -3\\0&0&-2
        \end{bsmallmatrix},~g_1\!=\!\begin{bsmallmatrix}
            0.5\\0\\0
        \end{bsmallmatrix},~g_2\!=\!\begin{bsmallmatrix}
            0\\0\\0.5
        \end{bsmallmatrix}.
\end{equation*}
Here, we assume that parameters $\delta_1$ and $\delta_2$ are unknown and time-varying, but bounded by $|\delta_1|\<0.007$ and $|\delta_2|\<0.015$. Selecting the nominal matrices $\bar A_j=A_j(0)$ and $\bar B_j=B_j(0),~\forall j\in\mathbb{K}$, and cycle $\nu=\{1,2\}$, Assumption~\ref{ass:nominal_system} is verified, yielding 
\begin{equation*}  
\begin{split}
P_1&=\begin{bsmallmatrix}0.73&3.77&-3.38\\3.778&25.7&-19.8\\-3.38&-19.8&28.49\end{bsmallmatrix},\ P_2=\begin{bsmallmatrix}4.51&4.89&-2.9\\4.89&6.31&-3.58\\-2.9&-3.58&4.02\end{bsmallmatrix},\\
\rho_1&=\begin{bsmallmatrix}3.84&
-0.65&
0.36\end{bsmallmatrix}^{\!\top},\quad\rho_2=\begin{bsmallmatrix}1.12&0.014&1.1042\end{bsmallmatrix}^{\!\top}
\end{split}
\end{equation*}
Considering that Assumption~\ref{ass:nominal_system} holds, we proceed to evaluate the robust stabilization conditions provided by Theorem~\ref{th:theo3}. To do so, setting $\mu=0.25$ and $\gamma=0.125$ and considering the above-computed vectors $\{\rho_i\}_{i\in\mathbb D_\nu}$ and matrices $\{P_i\}_{i\in\mathbb D_\nu}$, conditions \eqref{eq:condition} are verified using the CVX library in MATLAB, yielding
\begin{equation*}    R=\begin{bsmallmatrix}0.003&0.003&-0.0051\\0.003&0.0053&-0.0013\\-0.0051&-0.0013&0.0781\end{bsmallmatrix},\ Q=\begin{bsmallmatrix}65.8&7.56&-1.33\\7.56&158.7&-122.6\\-1.33&-122.6&553.01\end{bsmallmatrix}.
\end{equation*}
For initial conditions $x_0=\begin{bsmallmatrix}-1&1&-1\end{bsmallmatrix}^{\!\top}$ and $\sigma_0=u_0,$ where $u_0$ is a random value such that $u_0\in\{1,2\}$, the system response of the closed-loop system \eqref{eq:model_xitheta}-\eqref{eq:control} is depicted in Figure~\ref{fig:predictor}.
\begin{figure}[!t]   
     \centering
      \includegraphics[width=0.35\textwidth]{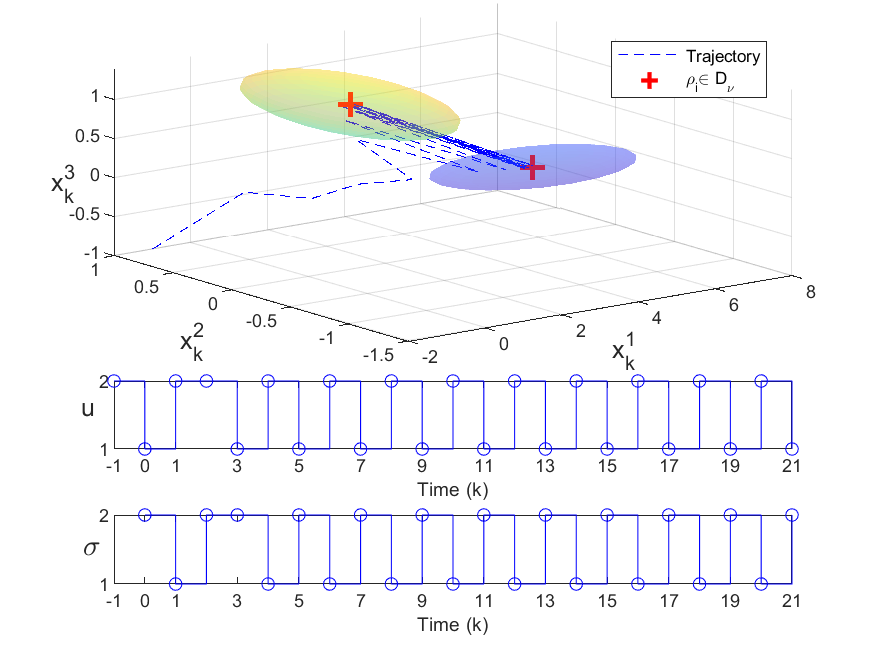}
      \vspace{-0.5cm}
    \caption{State trajectory of uncertain system \eqref{eq:model_xitheta} using the predictive switching control law \eqref{eq:control}.}
      \label{fig:predictor}
\end{figure}
This shows that the min-switching feedback control law \eqref{eq:control} compensates for the unitary input delay and system uncertainties, resulting in the robust exponential convergence of the system trajectory to a neighborhood surrounding the limit cycle $\{\rho_i\}_{i\in\mathbb D_\nu}$. Indeed, these neighborhoods can be estimated by the ellipsoids $\mathcal E(P_i-R,\rho_i)$ centered at the limit cycle $\{\rho_i\}_{i\in\mathbb{D}_{\nu}}$ (red crosses) as mentioned in Remark~\ref{remark:neighborhood}. In addition, note that the control law $u$ converges to the given cycle $\nu$, which is a direct consequence of  $\mathcal E(P_1-R,\rho_1)\bigcap \mathcal E(P_2-R,\rho_2)=\emptyset$, as reported in \citet{serieye2023attractors}. Finally, the evolution of the Lyapunov function \eqref{eq:LF} is depicted in Figure~\ref{fig:lyapunov_function}. This figure corroborates the invariance of attractor $\mathcal S_{\nu}$ $(V(\xi)\<1)$, which means that once the trajectory reaches $\mathcal S_{\nu}$, it never leaves, despite parameter uncertainties.
\begin{figure}[!t]   
     \centering
      \includegraphics[width=0.35\textwidth]{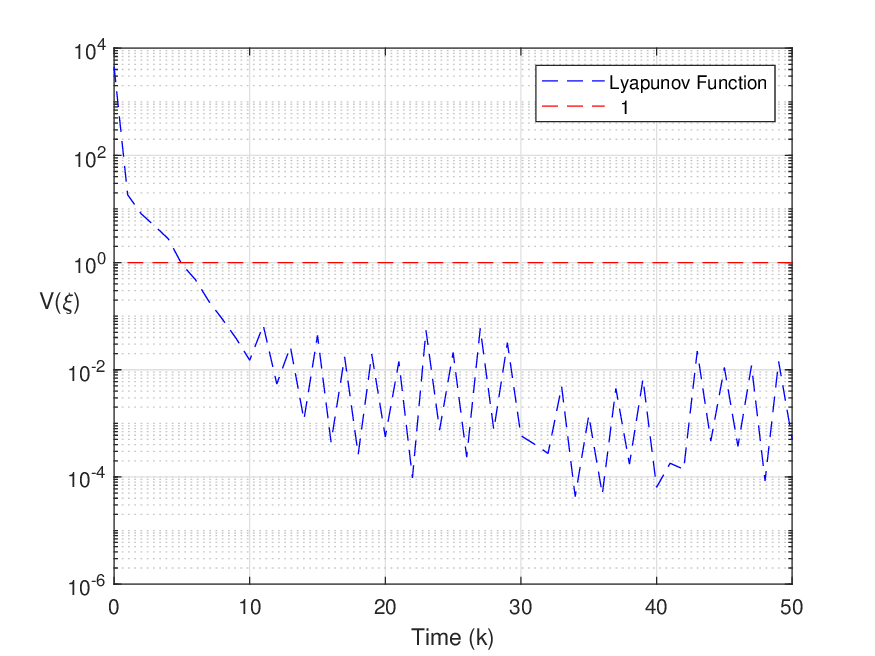}
      \vspace{-0.3cm}
    \caption{Lyapunov function of the uncertain system \eqref{eq:model_xitheta} with a log scale for discrete time.}
      \label{fig:lyapunov_function}
\end{figure}

\section{Conclusions}\label{sec:conclusions}
This paper presents robust stabilization conditions for uncertain switched systems with a unitary input delay. A prediction scheme based on nominal system parameters enables the design of the control law, compensating for both the input delay and parameter uncertainties. The conditions are derived using a Lyapunov function that characterizes the system attractor and a minimum-switching predictive control law,  ensuring exponential convergence to a robust limit cycle.
\\
Moreover, since the conditions depend only on the current mode and not on the previous one, this work lays the groundwork for extensions to a larger delay, for example, in power converter applications.

\bibliography{bib_SAS}           

\begin{thebibliography}{20}
\providecommand{\natexlab}[1]{#1}
\providecommand{\url}[1]{\texttt{#1}}
\providecommand{\urlprefix}{URL }
\expandafter\ifx\csname urlstyle\endcsname\relax
  \providecommand{\doi}[1]{doi:\discretionary{}{}{}#1}\else
  \providecommand{\doi}{doi:\discretionary{}{}{}\begingroup
  \urlstyle{rm}\Url}\fi

\bibitem[{Albea et~al.(2015)Albea, Garcia, and Zaccarian}]{albea2015hybrid}
Albea, C., Garcia, G., and Zaccarian, L. (2015).
\newblock Hybrid dynamic modeling and control of switched affine systems:
  application to {DC}-{DC} converters.
\newblock In \emph{54th IEEE Conference on Decision and Control}, 2264--2269.

\bibitem[{Beneux et~al.(2017)Beneux, Riedinger, Daafouz, and
  Grimaud}]{beneux2017robust}
Beneux, G., Riedinger, P., Daafouz, J., and Grimaud, L. (2017).
\newblock Robust stabilization of switched affine systems with unknown
  parameters and its application to {DC}/{DC} {F}lyback converters.
\newblock In \emph{2017 American Control Conference}, 4528--4533.

\bibitem[{Beneux et~al.(2019)Beneux, Riedinger, Daafouz, and
  Grimaud}]{beneux2019adaptive}
Beneux, G., Riedinger, P., Daafouz, J., and Grimaud, L. (2019).
\newblock Adaptive stabilization of switched affine systems with unknown
  equilibrium points: Application to power converters.
\newblock \emph{Automatica}, 99, 82--91.

\bibitem[{Bolzern and Colaneri(1988)}]{bolzern1988periodic}
Bolzern, P. and Colaneri, P. (1988).
\newblock The periodic {L}yapunov equation.
\newblock \emph{SIAM Journal on Matrix Analysis and Applications}, 9(4),
  499--512.

\bibitem[{Deaecto and Geromel(2017)}]{deaecto2016stability}
Deaecto, G. and Geromel, J. (2017).
\newblock Stability analysis and control design of discrete-time switched
  affine systems.
\newblock \emph{IEEE Trans. on Automatic Control}, 62(8), 4058--4065.

\bibitem[{Ebihara et~al.(2015)Ebihara, Arzelier, and Peaucelle}]{ebihara2015s}
Ebihara, Y., Arzelier, D., and Peaucelle, D. (2015).
\newblock \emph{S-variable approach to LMI-based robust control}.
\newblock Springer London.

\bibitem[{Egidio et~al.(2020)Egidio, Daiha, and Deaecto}]{egidio2020global}
Egidio, L., Daiha, H., and Deaecto, G. (2020).
\newblock Global asymptotic stability of limit cycle and
  $\mathcal{H}_2/\mathcal{H}_\infty$ performance of discrete-time switched
  affine systems.
\newblock \emph{Automatica}, 116, 108927.

\bibitem[{Egidio and Deaecto(2019)}]{egidio2019novel}
Egidio, L.N. and Deaecto, G.S. (2019).
\newblock Novel practical stability conditions for discrete-time switched
  affine systems.
\newblock \emph{IEEE Trans. on Automatic Control}, 64(11), 4705--4710.

\bibitem[{Mahmoud(2010)}]{mahmoud2010}
Mahmoud, M.S. (2010).
\newblock \emph{Switched Time-Delay Systems}.
\newblock Springer US.

\bibitem[{Mazenc et~al.(2017)Mazenc, Ahmed, and {\"O}zbay}]{mazenc2017state}
Mazenc, F., Ahmed, S., and {\"O}zbay, H. (2017).
\newblock State feedback stabilization of switched systems with delay:
  Trajectory based approach.
\newblock In \emph{2017 American Control Conference}, 4540--4543.

\bibitem[{Merch{\'a}n-Riveros et~al.(2024)Merch{\'a}n-Riveros, Albea, and
  Seuret}]{merchan2024data}
Merch{\'a}n-Riveros, M.C., Albea, C., and Seuret, A. (2024).
\newblock Data-driven control design for power converters approximated as
  switched affine systems and experimental validation.
\newblock \emph{IEEE Trans. on Circuits and Systems II: Express Briefs}.

\bibitem[{Portilla et~al.(2024)Portilla, Albea, and Seuret}]{portilla2024}
Portilla, G., Albea, C., and Seuret, A. (2024).
\newblock Predictive control design for discrete switched affine systems
  subject to a constant input delay.
\newblock \emph{IEEE Control Systems Letters}, 8, 2553--2558.

\bibitem[{Rubensson and Lennartson(2000)}]{rubensson2000stability}
Rubensson, M. and Lennartson, B. (2000).
\newblock Stability of limit cycles in hybrid systems using discrete-time
  {L}yapunov techniques.
\newblock In \emph{39th IEEE Conference on Decision and Control}, 1397--1402.

\bibitem[{Rubensson et~al.(1998)Rubensson, Lennartson, and
  Pettersson}]{rubensson1998convergence}
Rubensson, M., Lennartson, B., and Pettersson, S. (1998).
\newblock Convergence to limit cycles in hybrid systems-an example.
\newblock \emph{IFAC Proceedings Volumes}, 31(20), 683--688.

\bibitem[{Sanchez et~al.(2019)Sanchez, Polyakov, Fridman, and
  Hetel}]{sanchez2019switching}
Sanchez, T., Polyakov, A., Fridman, E., and Hetel, L. (2019).
\newblock A switching controller for a class of {MIMO} bilinear systems with
  time delay.
\newblock \emph{IEEE Trans. on Automatic Control}, 65(5), 2250--2256.

\bibitem[{Scherer(2001)}]{scherer2001theory}
Scherer, C. (2001).
\newblock Theory of robust control.
\newblock \emph{Delft University of Technology}, 1--160.

\bibitem[{Serieye et~al.(2023)Serieye, Albea, Seuret, and
  Jungers}]{serieye2023attractors}
Serieye, M., Albea, C., Seuret, A., and Jungers, M. (2023).
\newblock Attractors and limit cycles of discrete-time switching affine
  systems: nominal and uncertain cases.
\newblock \emph{Automatica}, 149, 110691.

\bibitem[{Sun and Ge(2011)}]{sun2011stability}
Sun, Z. and Ge, S.S. (2011).
\newblock \emph{Stability theory of switched dynamical systems}.
\newblock Springer London.

\bibitem[{Theunisse et~al.(2015)Theunisse, Chai, Sanfelice, and
  Heemels}]{theunisse2015robust}
Theunisse, T.A., Chai, J., Sanfelice, R.G., and Heemels, W.M.H. (2015).
\newblock Robust global stabilization of the {DC}-{DC} boost converter via
  hybrid control.
\newblock \emph{IEEE Trans. on Circuits and Systems I: Regular Papers}, 62(4),
  1052--1061.

\bibitem[{Vu and Morgansen(2010)}]{vu2010stability}
Vu, L. and Morgansen, K.A. (2010).
\newblock Stability of time-delay feedback switched linear systems.
\newblock \emph{IEEE Trans. on Automatic Control}, 55(10), 2385--2390.

\end{thebibliography}
\end{document}